\documentclass[10pt]{article}
\font\smallit=cmti10

\usepackage{caption}
\usepackage{color}
\usepackage{amssymb,latexsym,amsmath,epsfig,amsthm} 
\usepackage{empheq}
\usepackage{here}
\usepackage{url}
\usepackage{ascmac}
\makeatletter
\usepackage{here}
\usepackage{comment}

\renewcommand\section{\@startsection {section}{1}{\z@}
{-30pt \@plus -1ex \@minus -.2ex}
{2.3ex \@plus.2ex}
{\normalfont\normalsize\bfseries\boldmath}}

\renewcommand\subsection{\@startsection{subsection}{2}{\z@}
{-3.25ex\@plus -1ex \@minus -.2ex}
{1.5ex \@plus .2ex}
{\normalfont\normalsize\bfseries\boldmath}}

\renewcommand{\@seccntformat}[1]{\csname the#1\endcsname. }

\makeatother
\newtheorem{theorem}{Theorem}
\newtheorem{lemma}{Lemma}

\newtheorem{proposition}{Proposition}

\theoremstyle{definition}
\newtheorem{definition}{Definition}

\begin{document}

\begin{center}
\uppercase{\bf A New 0(\lowercase{k} $\mathbf{\lowercase{log x}}$) Algorithm for Josephus Problem}
\vskip 20pt
{\bf Hikaru Manabe}\\
{\smallit Keimei Gakuin Junior and High School, Kobe City, Japan}\\
{\tt urakihebanam@gmail.com}
\vskip 10pt
{\bf Ryohei Miyadera }\\
{\smallit Keimei Gakuin Junior and High School, Kobe City, Japan}\\
{\tt runnerskg@gmail.com}
\vskip 10pt
{\bf Yuji Sasaki}\\
{\smallit Graduate School of Advanced Science and Engineering Hiroshima University, Higashi-Hiroshima City, Japan}\\
{\tt urakihebanam@gmail.com}
\vskip 10pt
{\bf Shoei Takahashi}\\
{\smallit Keio University Faculty of Environment and Information Studies, Fujisawa City, Japan}\\
{\tt shoei.takahashi@keio.jp}
\vskip 10pt
{\bf Yuki  Tokuni}\\
{\smallit Graduate School of Information Science University of Hyogo, Kobe City, Japan}\\
{\tt ad21o040@gsis.u-hyogo.ac.jp}
{\tt \ }


\end{center}

\centerline{\bf Abstract}
\noindent
Let $k,n$ be natural numbers such that $n,k \geq 2$.
We present a $O(k \log n)$ algorithm to determine the last number that remains in 
the Josephus problem, where numbers $1,2,\dots, n$ are arranged in a circle and every $k$-th numbers are removed.
The algorithm is as the followings.\\
$(i)$ We start with $x=0$. \\
$(ii)$ We substitute $x$ with $x +  \left\lfloor \frac{x}{k-1} \right\rfloor +1$ until we get $x$ such 
$nk-n \leq x$. \\
$(iii)$ Then, $nk-x$ is the number that remains.

The time complexity of our algorithm is $O(k \log n)$, and this time complexity is on a par with the existing $O(k \log n)$ algorithm.  We do not have any recursion overhead or stack overflow because we do not use any recursion. Therefore, the space complexity of our algorithm is $O(1)$, and ours is better than the existing $O(k \log n)$ algorithm in this respect.
When $k$ is small and $n$ is large, our algorithm is better than the existing $O(k \log n)$ algorithm. 

\section{Introduction}
Let $\mathbb{Z}_{\ge 0}$ and $\mathbb{N}$ represent sets of non-negative and positive integers, respectively.
Let $n,k \in \mathbb{N}$.
We have a finite sequence of positive integers $1,2,3$
$, \cdots, n-1, n$ arranged in a circle, and we start to remove every $k$-th number until only one remains. The Josephus problem is to determine the number that remains.

Many researchers have tried to find an efficient algorithm to determine the number that remains in the Josephus problem.

D. Knuth \cite{knuth} proposed $O(n \log n)$ algorithm, and E. Lloyd \cite{O(nlogm)} proposed $O(n log k)$ algorithm for the case that $k < n$.
For a small $k$ and a large $n$, a $O(k \log n)$ algorithm is proposed in \cite{josephusalgorithm}.

Our algorithm makes use of maximum Nim that is a topic of combinatorial game theory. 

There will be three different ways to read the present article.

If you are a mathematician who wants to understand our algorithm mathematically, please read our article \cite{integer2024}  that will be published soon. Then,
please read Section \ref{sectioforalgo} of the present article.

If you want to understand the mathematical background of the algorithm without going into the details of proofs, please read the present article from the first section to the last section. We provided proof for Theorem \ref{theofjosemaxcal}, but we omitted proofs for other theorems and lemmas. Without the proof, it won't be easy to understand Theorem \ref{theofjosemaxcal}.

If you just want to know our algorithm without any mathematical background, please read Section \ref{sectioforalgo} after skipping Section \ref{maxnimjose}.

\section{Maximum Nim with  the Rule Function $f(x)=\left\lfloor \frac{x}{k} \right\rfloor$ for a Positive Integer Such That $k \geq 2$ and the Josephus Problem }\label{maxnimjose}
For any real number $x$, the floor of $x$ denoted by  $ \left\lfloor x \right\rfloor$ represents the greatest integer less than or equal to $x$.

The classic game of Nim is played using stone piles. Players can remove any number of stones from any pile during their turn, and the player who removes the last stone is considered the winner. Several variants of the classical Nim game exist. For the Maximum Nim, we place an upper bound $f(n)$ on the number of stones that can be removed in terms of the number $n$ of stones in the pile (see \cite{levinenim}). For articles on Maximum Nim, see \cite{thaij2023b} and  \cite{integer2023}.

 This study explores the relation between a Maximum Nim and the Josephus problem. Let  $ f(x) = \left\lfloor \frac{x}{k} \right\rfloor$, where $k$ is a positive integer such that $k \geq 2$.

We define Maximum Nim with the rule function $f$.
\begin{definition}
 Suppose there is a pile of $n$ stones and two players take turns to remove stones from the pile.
At each turn, the player can remove at least one and at most 
$f(m)$ stones if the number of stones is $m$. The player who cannot make another move loses the game.   We call $f(m)$ the rule function of the game.
\end{definition}

We prove that there is a simple relation between a Maximum Nim with the rule function $f$ and the Josephus problem, in which every $k$-th number is to be removed.
 This is remarkable because the games for Nim and Josephus problems are considered entirely different.

\begin{definition}
 We denote the pile of $m$ stones as $(m)$, which we refer to as the \textit{position} of the game.
\end{definition}

We briefly review some necessary concepts in combinatorial game theory (see \cite{lesson} for more details). We consider impartial games with no draws, and therefore, there are only two outcome classes. \\
$(a)$ A position is called a $\mathcal{P}$-\textit{position}, if it is a winning position for the previous player, as long as he/she plays correctly at every stage.\\
$(b)$ A position is called an $\mathcal{N}$-\textit{position}, if it is a winning position for the next player, as long as he/she plays correctly at every stage.

The Grundy number is one of the most important tools in research on combinatorial game theory, and we define it in the following definition.
\begin{definition}\label{defofmexgrundy}
$(i)$ For any position $(x)$, there exists a set of positions that can be reached in precisely one move in this game, which we denote as \textit{move}$(x)$.\\
$(ii)$ The \textit{minimum excluded value} $(\textit{mex})$ of a set $S$ of non-negative integers is the smallest non-negative integer that is not in S. \\
$(iii)$ Let $(x)$ be a position in the game. The associated \textit{Grundy number} is denoted by $\mathcal{G}(x)$ and is recursively defined as follows:
$\mathcal{G}(x) = \textit{mex}(\{\mathcal{G}(u): (u) \in move(x)\}).$
\end{definition}

For the maximum Nim of $x$ stones with rule function $f(x)$, \\
$\textit{move}(x)$ $= \{x-u:1 \leq u \leq f(x) \text{ and } u \in \mathbb{N} \}$.

We assume that  $\mathcal{G}$ is the Grundy number of the Maximum Nim with rule function $f(x)=\left\lfloor \frac{x}{k} \right\rfloor$. The next result demonstrates the usefulness of the Sprague--Grundy theorem for impartial games.
\begin{theorem}\label{theoremoforgrundyn}
For any position $(x)$,
	$\mathcal{G}(x)=0$ if and only if $(x)$ is the $\mathcal{P}$-position.
\end{theorem}
See \cite{lesson} for the proof of this theorem.

We define the Josephus problem. For the details of the Josephus problem, see \cite{computerscience}.
\begin{definition}\label{defofjj}
Let $n,k \in \mathbb{N}$.
We have a finite sequence of positive integers $1,2,3$
$, \cdots, n-1, n$ arranged in a circle, and we start to remove every $k$-th number until only one remains.
This is the well-known Josephus problem. For $m$ such that $1 \leq m \leq n$ and $1 \leq i \leq n-1$, we define
\begin{equation}
o_n(m)=i \nonumber
\end{equation}
where $m$ is the $i$-th number to be removed, and
\begin{equation}
o_n(m)=n \nonumber
\end{equation}
where $m$ is the number that remains after removing other $n-1$ numbers.
\end{definition}

The following theorem presents a simple relation between the Maximum Nim and the Josephus problem.
\begin{theorem}\label{grundytheoremjj}
For any $n,m \in \mathbb{N}$ such that $1 \leq m \leq n$, we have 
\begin{equation}
n-o_n(m) = \mathcal{G}(nk-m). \label{theoremforgj}
\end{equation}
\end{theorem}
For a proof of this lemma, see Theorem 2 of 
\cite{integer2024}.

\begin{definition}\label{seconddefn}
For $k \in \mathbb{N}$ such that $k \geq 2$,
we define a function $h_{k}(x)$ for a non-negative integer $x$ as
\begin{equation}
h_{k}(x)= x + \left\lfloor \frac{x}{k-1} \right\rfloor +1. \nonumber  
\end{equation}
\end{definition}

\begin{lemma}\label{seconddefnb}
For the function $h_{k}(x)$ in Definition \ref{seconddefn}, the following hold:\\
$(i)$ $h_{k}(x) < h_{k}(x^{\prime})$ for any $x, x^{\prime} \in \mathbb{Z}_{\ge 0}$ such that $x < x^{\prime}$;\\
$(ii)$ $x < x+1 \leq h_{k}(x)$ for any $x \in \mathbb{Z}_{\ge 0}$;\\
$(iii)$  $\lim\limits_{p\to\infty}  h_{k}^p(x_0)= \infty$, where $h_{k}^p$ is the $p$-th functional power of $h_k$. 
\end{lemma}
For a proof of this lemma, see Lemma 6 of \cite{integer2024}.

\begin{proposition}\label{theoremforgk}
For $h_{k}$ in Definition \ref{seconddefn}, the following hold: \\
$(i)$	$h_{k}(x)= \left\lfloor \frac{h_{k}(x)}{k} \right\rfloor + x+1$ for $x \in \mathbb{Z}_{\ge 0}$;\\
$(ii)$	$\mathcal{G}(h_{k}(x)) = \mathcal{G}(x)$ for $x \in \mathbb{Z}_{\ge 0}$;\\
$(iii)$ For any $m \in \mathbb{Z}_{\ge 0}$, there exists $x_{0} \in \mathbb{Z}_{\ge 0}$ such that 
\begin{equation}
\{x \in \mathbb{Z}_{\ge 0}: \mathcal{G}(x)=m \} = \{h_{k}^p(x_{0}):p \in \mathbb{Z}_{\ge 0}\}.\nonumber 
\end{equation}
\end{proposition}
For a proof of this lemma, see Proposition 1 of 
\cite{integer2024}.

\begin{theorem}\label{theofjosemaxcal}
Let $n$ be a natural number. Then, there exists $p  \in \mathbb{N}$ such that 
$h_{k}^{p-1}(0) < n(k-1) \leq h_{k}^{p}(0)$, and the last number that remains is $nk-h_{k}^{p}(0)$ in the Josephus problem of $n$ numbers, where every $k$-th number is removed.
\end{theorem}
\begin{proof}
Let $m$ be the number that remains in the Josephus problem of $n$ numbers, where the $k$-th number is removed. Then, based on Definition \ref{defofjj} and Theorem \ref{grundytheoremjj}, 
\begin{equation}
\mathcal{G}(nk-m)=0. \label{mathcalgnk0} 
\end{equation}
Since $h^0_k(0)=0$, by $(ii)$ and $(iii)$ of Lemma \ref{seconddefnb}, there exists $p \in \mathbb{N}$ such that 
\begin{equation}
h_{k}^{p-1}(0) < n(k-1) \leq h_{k}^{p}(0).\label{hkp1nk1hkp}
\end{equation}
Since $\mathcal{G}(0)=0$, by Proposition \ref{theoremforgk}, 
\begin{equation}
\mathcal{G}(h_{k}^{p-1}(0))=\mathcal{G}(h_{k}^{p}(0))=\mathcal{G}(h_{k}^{p+1}(0))=\mathcal{G}(0)=0.\nonumber
\end{equation}
As
\begin{equation}
h_k(n(k-1))=n(k-1)+n+1 = nk+1,\nonumber
\end{equation}
by Inequality (\ref{hkp1nk1hkp}) and $(i)$ of Lemma \ref{seconddefnb},
\begin{align}
h_k^{p+1}(0) & = h_k(h_k^{p}(0)) \nonumber \\
& \geq h_k(n(k-1))  \nonumber \\
& =  nk+1. \nonumber
\end{align}
Hence, by Inequality (\ref{hkp1nk1hkp}), 
\begin{equation}
h_k^{p-1}(0) < n(k-1) \leq nk-m < nk+1 \leq h_k^{p+1}(0).\label{nk1nkmhkp10}
\end{equation}
By $(ii)$ of Lemma \ref{seconddefnb},
\begin{equation}
h_{k}^{p-1}(0) <  h_{k}^{p}(0) < h_{k}^{p+1}(0),\nonumber
\end{equation}
and hence by Equation (\ref{mathcalgnk0}), Inequality (\ref{nk1nkmhkp10}) and (iii) of Proposition \ref{theoremforgk}, 
$nk-m = h_{k}^{p}(0)$. Hence, $m = nk-h_{k}^{p}(0)$.
\end{proof}


\section{A New 0($k\log x)$ Algorithm for Josephus Problem}\label{algorithm}
We use the following formula in this section. This formula is presented in Theorem \ref{theofjosemaxcal} of the previous section.

Let $n,k$ be a natural number such that $k \geq 2$ and 
\begin{equation}
h_{k}(x)= x + \left\lfloor \frac{x}{k-1} \right\rfloor +1.\label{definitionoffunctionhk}
\end{equation}
If 
\begin{equation}
h_{k}^{p-1}(0) < n(k-1) \leq h_{k}^{p}(0)  \label{betweentwoterm}  
\end{equation}
for some $p \in \mathbb{N}$, the last number that remains is 
\begin{equation}
nk-h_{k}^{p}(0)\label{thelastnumber}
\end{equation}
in the Josephus problem of $n$ numbers, where every $k$-th number is removed.

Based on (\ref{definitionoffunctionhk}), (\ref{betweentwoterm}) and (\ref{thelastnumber}),
we build an algorithm for the Josephus problem.\\\\
$(i)$ We start with $x=0$. \\
$(ii)$ We substitute $x$ with $x +  \left\lfloor \frac{x}{k-1} \right\rfloor +1$ until we get $x$ such 
$nk-n \leq x$. This process is based on (\ref{definitionoffunctionhk}) and (\ref{betweentwoterm}).\\
$(iii)$ Then, we get $m=nk-x$, and $m$ is the number that remains.

The time complexity of our algorithm is $O(k \log n)$, and this time complexity is on a par with a $O(k \log n)$ algorithm in \cite{josephusalgorithm}.  We do not have any recursion overhead or stack overflow because we do not use any recursion. Therefore, the space complexity of our algorithm is $O(1)$, and ours is better than the existing $O(k \log n)$ algorithm in this respect.

The structure of our algorithm is more simple than 
the $O(k \log n)$ algorithm in \cite{josephusalgorithm}.

The following is the Python implementation of our algorithm. We have $n$ numbers 
$0,1,2,\dots, n-1$ and remove every $k$th number. We get the number that remains after the calculation.
\begin{verbatim}
def proposed_algorithm(n, k):
  x = 0   # (1)
  range_min = (k - 1) * n 
  while x < range_min:   # (2)
    x += x // (k - 1) + 1   # (3)
  m = n * k - x   # (4)
  return  m - 1   # (5)
\end{verbatim}

\section{Comparison between Our Algorithm and Existing Algorithms}\label{sectioforalgo}
Here, we compare our algorithm to various existing algorithms. 

We present two existing algorithms and compare our algorithm to these.
\subsection{An Existing $O(n)$ Algorithm}\label{O(n)}
We have $n$ numbers 
$0,1,2,\dots, n-1$ and remove every $k$th number. We get the number that remains after the calculation.
In the following Python implementation, the algorithm is the same as the definition of the Josephus problem itself.
Here, we do not use recursion to avoid stack overflow.
\begin{verbatim}
def benchmark_b(n, k):
    if n == 1:
        return 0
    result = 0
    for i in range(2, n + 1):
        result = (result + k) % i
    return result
\end{verbatim}

$result = (result + k) \% i $
Since this algorithm does not use recursion, the space complexity is $0(1)$.Since time complexity is $0(n)$, the running time will increase rapidly as $n$ increases.

\subsection{An Existing $O(k \log n)$ Algorithm }\label{O(klogn)}
The following algorithm is based on the method of removing $k$-th, $2k$-th, ..., $\left\lfloor \frac{n}{k} \right\rfloor k$-
the number as one step. This algorithm used $cnt = n // k$ to remove these numbers as one step.

\begin{verbatim}
	def benchmark_b(n, k):
  if n== 1:
    return 0
  if k == 1:
    return n - 1
  if n < k:
    return (benchmark_b(n - 1, k) + k) % n
  
  cnt = n // k
  result = benchmark_b(n - cnt, k)
  result -= n % k
  if result < 0:
    result += n
  else:
    result += result // (k - 1)

  return result
\end{verbatim}
The time complexity of this algorithm is $O(k \log n)$, however, the space complexity is $O(n)$, and we may have 
the stack overflow because of the use of recursion.
\subsection{Comparison of running time of each algorithm}
We compare our algorithm to two existing algorithms by program execution time.
We use a MacBook Pro(CPU:M1 Max, RAM:32GB), and we run each program $3000$ times.

In the following statement, we denote the execution time of the algorithm in Subsection \ref{O(n)} by Benchmark A,
and the execution time of the algorithm in Subsection \ref{O(klogn)} by Benchmark B.

In Graph \ref{graph_fixed_k}, we let $k=200$ and we increase $n$ from $1$ to $5000$ ($step=100$).
When $n$ is small, Benchmark A,b are small, but Benchmark A increases rapidly because of its time complexity $O(n)$ as $n$ increases.
When $n$ is large,  the execution time of our algorithm and Benchmark B are smaller than Benchmark A. However, if we compare the execution time of our algorithm to Benchmark B, ours is better.
The reason is that ours has a smaller space complexity, although both have the same time complexity $O(k\log n)$.
Therefore, when $n$ is large, our algorithm is better than these two algorithm.

\begin{figure}[H]
\begin{center}
\includegraphics[height=4cm]{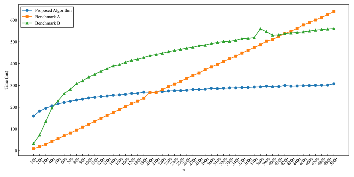}
\caption{ $k=200$}
\label{graph_fixed_k}
\end{center}
\end{figure}

In Graph \ref{graph_fixed_n}, we let $n=300$, and increase $k$ from $10$ to $300$ ($step=10$).
Benchmark A is constant regardless of the size of $k$. The execution time of our algorithm and Benchmark B
increase as $k$ increases, however, Benchmark B does not increase endlessly because it does not depend on $k$ when $n<k$.

The efficiency of our algorithm, the algorithm in Subsection \ref{O(n)}, and the algorithm in Subsection \ref{O(klogn)} depend on $k$ and $n$.

When $k$ is small and $n$ is large, our algorithm and the algorithm in Subsection \ref{O(klogn)} are better than the algorithm in Subsection \ref{O(n)}. We can explain this fact by the time complexity $0(k \log n)$, but our algorithm is better than the algorithm in Subsection \ref{O(klogn)}.
When $k$ is large and $n$ is not large, the algorithm in Subsection \ref{O(n)} is better than others.

\begin{figure}[H]
\begin{center}
\includegraphics[height=4cm]{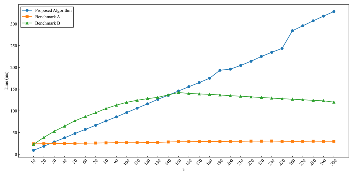}
\caption{$n=300$}
\label{graph_fixed_n}
\end{center}
\end{figure}




\end{document}